\documentclass[10pt,a4paper,twoside]{article}

	\usepackage[hyphens]{url}				
	\usepackage{hyperref}									
		\hypersetup{colorlinks,linkcolor={green!30!black},citecolor={blue!50!black},urlcolor={blue!80!black}}
		\hypersetup{breaklinks=true}		
	\usepackage{mathtools,amssymb,amsthm}						
	\usepackage{dsfont}										
	\usepackage[
		backend=biber,style=authoryear-icomp,natbib=true,
		sortcites=true,ibidtracker=strict,sorting=nyvt,giveninits=true,
		uniquename=init,maxbibnames=99]{biblatex}				
		\renewbibmacro{in:}{}								
			\newbibmacro{string+doi}[1]{%
  			\iffieldundef{doi}{#1}{\href{http://dx.doi.org/\thefield{doi}}{#1}}}
			\DeclareFieldFormat{title}{\usebibmacro{string+doi}{\mkbibemph{#1}}}
			\DeclareFieldFormat[article]{title}{\usebibmacro{string+doi}{\mkbibquote{#1}}}
	
	\usepackage[hyphenbreaks]{breakurl}		
	\usepackage[
		top=3.1cm,bottom=3.1cm,
		inner=2.5cm,outer=2.5cm,
		a4paper]{geometry}									
	\usepackage{tikz}
	\usepackage[font=small,labelfont=bf,centerlast]{caption}			
	\usepackage{array}										
		\newcolumntype{L}{>{$}l<{$}} 							
		\newcolumntype{C}{>{$}c<{$}} 						
		\newcolumntype{R}{>{$}r<{$}} 							
	\usepackage[shortlabels]{enumitem}								
	
	\makeatletter
		\AtBeginDocument{\toggletrue{blx@useprefix}}
		\AtBeginBibliography{\togglefalse{blx@useprefix}}
	\makeatother

	\newcommand{\ncom}{\newcommand}
	\ncom{\h}{\mathcal{H}}
	\ncom{\een}{\mathds{1}}
	\ncom{\ket}[1]{\left|#1\right\rangle}
	\ncom{\bra}[1]{\left\langle#1\right|}
	\ncom{\braket}[2]{\left\langle#1\middle|#2\right\rangle}
	\ncom{\ketbra}[2]{\left|#1\middle\rangle\middle\langle#2\right|}
	\ncom{\expv}[3]{\left\langle#1\middle|#2\middle|#3\right\rangle}
	\ncom{\bket}[6]{\left|#1,#2\middle|#3,#4\middle|#5,#6\right\rangle}
	\ncom{\proj}[6]{\left\lgroup#1,#2\middle|#3,#4\middle|#5,#6\right\rgroup}
	\ncom{\set}[2]{\left\{#1\:\middle|\:#2\right\}}				
	\ncom{\dee}{\mathrm{d}}								
	\ncom{\ldef}{\coloneqq}								
	\ncom{\rdef}{\eqqcolon}								
	\ncom{\Epsilon}{\mathcal{E}}							
	\ncom{\pee}{\mathds{P}}								
	\ncom{\peetje}{\mathrm{p}}							
	\ncom{\ee}{\mathds{E}}								
	\ncom{\var}{\mathds{V}\mathrm{ar}}					
	\ncom{\cov}{\mathds{C}\mathrm{ov}}					

	\theoremstyle{plain}	
		\newtheorem{theorem}{Theorem}
		\newtheorem*{theorem*}{Theorem}
		
		\newtheorem{lemma}{Lemma}
		\ncom{\claimautorefname}{Claim}	
	\theoremstyle{definition}
		\newtheorem{definition}{Definition}\ncom{\definitionautorefname}{Definition}	
	\def\hyph{-\penalty0\hskip0pt\relax}

\begin{document}

\title{How Real are Quantum States in $\psi$-ontic Models?} 
\author{R. Hermens}

\maketitle

\begin{abstract}
There is a longstanding debate on the metaphysical relation between quantum states and the systems they describe.
A series of relatively recent $\psi$-ontology theorems have been taken to show that, provided one accepts certain assumptions, ``quantum states are real''.
In this paper I investigate the question of what that claim might be taken to mean in light of these theorems.
It is argued that, even if one accepts the framework and assumptions employed by such theorems, such a conclusion is not warranted.
Specifically, I argue that when a so-called ontic state is taken to describe the properties of a system, the relation between this state and some quantum state as established by $\psi$-ontology theorems, is not of the kind that would warrant counting the quantum state among these properties in any way.
\end{abstract}

\tableofcontents

\section{Introduction}
This paper is about the central question ``Is the quantum state real?''. 
It is a question that presumably arose in tandem with the first axiomatizations of quantum mechanics in the early twentieth century.
About as old is the more important question of what that central question means.
In recent years, we find ourselves in the peculiar situation that despite a lack of consensus on what the answer to this second question should be, there is some consensus that the answer to the central question should be ``yes'' (provided one buys into some assumptions).

At first sight, this seems as unhelpful as giving the same ``yes''-answer to the question ``particle or wave?''.
But in the present situation, the ``yes'' comes in the form of a series of formal results known as $\psi$-ontology theorems.
The most famous of these results is the PBR Theorem, named after \citet{PBR12}.
Being formal results, they can only give an answer to a specific well-defined reading of the central question.
In this paper I clarify what this reading should be taken to be.
More accurately, I shall argue that, somewhat provocatively, $\psi$-ontology theorems have little to say about the metaphysical status of quantum states.
Instead, they are better understood as theorems ruling out certain specific explanations of some of the phenomena described by quantum mechanics.

My criticism should not be interpreted as an endorsement of an epistemic interpretation of quantum states.
Nor do I wish to claim that there are no good arguments in favor of an ontic interpretation (see e.g. \citep{Brown19}).
Instead, my aim is to clarify the current import of $\psi$-ontology theorems on this debate.
It is my view that such clarification is to be welcomed, as there seems to be some confusion about it going around.
This is not in the least place due to the peculiar circumstances in which the PBR Theorem was introduced to the world.
The original preprint title of the paper was ``The quantum state cannot be interpreted statistically''.
It was this ambiguous statement that was quickly hyped with soundbites like ``This strips away obscurity and shows you can't have an interpretation of a quantum state as probabilistic'' (Wallace, as quoted in \citep{Reich11}).
Surely, what isn't meant here is that quantum states do not give rise to probabilities.

Only a few months later, confusion was added for those not following the developments closely, when \citet{Lewis12} put online a preprint of their paper with the title ``The quantum state \emph{can} be interpreted statistically'' (emphasis added).
Not long after that, the situation was mediated somewhat when Pusey, Barrett and Rudolph gave a modified characterization of their theorem as establishing that, under certain circumstances, ``quantum states must be real''.
More recently, \citet{Wallace18} gave a more philosophically sounding characterization of what $\psi$-ontology theorems intend to show: the necessity of ``Quantum state representationalism: The distinct states in quantum state space represent different physical possibilities.''
In other words, one physical possibility cannot accurately be represented by two distinct quantum states.

There is of course a remaining ambiguity concerning what is meant with a ``physical possibility'' and how \emph{that} is to be represented by the theory.
In $\psi$-ontology theorems, it is assumed that these can be captured in terms of so-called ontic states (see Section \ref{OnticSec}).
Clearly, Wallace takes this to be a reasonable assumption as he writes about $\psi$-ontology theorems that ``it looks reasonably clear (without being \emph{universally} accepted) that any plausible non-representational reading of the quantum state will have to presume instrumentalism \citep{FuchsPeres00} or some other radical departure from the usual scientific-realist conception of physical theories as giving a third-party, agent-independent account of the world''.

\citet{Leifer14} adopts a similar view when discerning ``realist'' $\psi$-epistemic views and ``neo\hyph Copenhagen'' $\psi$-epistemic views, and clarifying that $\psi$-ontology theorems only rule out the former.
\citet{BenMenahem17} in turn argues that these theorems should be interpreted as supporting a neo\hyph Copenhagen type interpretation for those preferring a $\psi$-epistemic view, but denies that this means going as far as adopting a kind of anti-realist stance towards quantum mechanics as a whole. 
\Citet{OldofrediLopez20} argue that the framework adopted in $\psi$-ontology theorems is actually not adequate to fully capture the realist/anti-realist dichotomy for quantum states. 
Finally, \citet{Halvorson19} argues that the realist/anti-realist division for quantum states isn't really as unambiguous as one might hope at all.

The combined impression these papers leave us with, is that the precise implications of $\psi$-ontology theorems are not well understood, or at least up for debate.
In this paper I will steer away somewhat from the metaphysical discussion of when exactly quantum state representationalism holds or one is justified in saying that ``the quantum state is real''.
Instead, my focus is on what kind of partial answers to the central question $\psi$-ontology theorems intend to deliver, and what partial answers can actually be inferred from these theorems.
So it is the formal notion of quantum state representationalism as adopted in $\psi$-ontology theorems that I will investigate and criticize, but I will not go into the question of how quantum state representationalism should be given a formal account instead.

Like any formal result, $\psi$-ontology theorems have to presuppose a certain mathematical framework within which then certain philosophically or physically motivated assumptions can be formalized to play a role in the proof for the claim to be made.
I shall not be concerned with the reasonableness of this framework or the adopted auxiliary assumptions in this paper.
But understanding the framework is essential for understanding my critical stance towards $\psi$-ontology theorems.
It will be introduced and discussed in Section \ref{OnticSec}, where I will also explain that the interpretation of $\psi$-ontology theorems is tied up with the interpretation of these frameworks.
Then, in Section \ref{ReducSec}, I will explain that, insofar as $\psi$-ontology theorems prove the necessity of quantum state representationalism, this cannot be taken to imply that quantum states themselves are part of what constitutes a ``physical possibility''.
In Section \ref{NotShowSec} it is argued that $\psi$-ontology theorems do not even establish quantum state representationalism, even if one accepts the framework and auxiliary assumptions they adopt.
To this end I also discuss two explicit ontic models.
The first is due to \citet{Gudder70} and the second due to \citet{Meyer99}, \citet{Kent99} and \citet{CliftonKent00}.
Although these models are shown to be $\psi$-ontic in the sense adopted in $\psi$-ontology theorems, there is no way of unambiguously linking quantum states to what constitutes a physical possibility in these models.
I end on a more positive note in Section \ref{DoShowSec} with a discussion of what I do take $\psi$-ontology theorems to show.


\section{Using Ontic Models}\label{OnticSec}
To be able to prove anything about the ``real properties'' of a system, one needs some mathematical object that is taken to represent these properties.
To this end the notion of an \emph{ontic state} $\lambda$ of a system is introduced. 
Very little is assumed about what kind of objects ontic states are, other than that the set of all possible ontic states $\Lambda$ forms a measurable space, i.e., it comes equipped with a $\sigma$-algebra $\Sigma$ of subsets of $\Lambda$.
This ensures that the introduction of probabilities is well-defined.

How exactly $\lambda$ represents properties of a system is left open.
The minimal condition is merely that $\lambda$ can play the functional role for making predictions about the outcomes of future measurements.
That is, for every physical quantity $A$, every ontic state $\lambda$ determines the probability $\peetje_A(a|\lambda)$ for every possible outcome $a$ of a measurement of $A$.
Formally, $\peetje_A$ is a Markov kernel from $(\Lambda,\Sigma)$ to the measurable space of possible measurement outcomes for $A$.\footnote{This means that the map $\lambda\mapsto\peetje_A(\Delta|\lambda)$ is a measurable function for every measurable set of possible outcomes $\Delta$ and $\Delta\mapsto\peetje_A(\Delta|\lambda)$ is a probability measure for every $\lambda$.}
How these measurement outcomes come about is left unspecified: ontic models need not pose a solution to the measurement problem.

Ontic models are required to (partially) reproduce the predictions of quantum mechanics.
In this paper I shall look specifically at fragments of quantum mechanics (see also \citep[\S4]{Leifer14}).
A \emph{fragment} of quantum mechanics consists of a triplet $(\h,\mathcal{P},\mathcal{M})$, where $\h$ is a finite-dimensional Hilbert space, $\mathcal{P}$ is a set of density operators, and $\mathcal{M}$ is a set of POVMs.
An ontic model for a fragment $(\h,\mathcal{P},\mathcal{M})$ consists of a measurable space $(\Lambda,\Sigma)$ and further satisfies the following two conditions. 
\begin{enumerate}[(I)]
\item For every physical quantity $A$ that can be represented by a POVM $\{E_{a_1,},\ldots,E_{a_n}\}\in\mathcal{M}$, there is a Markov kernel $\peetje_A$ such that $$\sum_i\peetje_A(a_i|\lambda)=1$$ for every $\lambda\in\Lambda$ and
\item For every density operator $\rho\in\mathcal{P}$ there is a no-empty set $\Pi_\rho$ of probability distributions over $(\Lambda,\Sigma)$ such that $$\int\peetje_A(a_i|\lambda)\dee\mu(\lambda)=\mathrm{Tr}(\rho E_{a_i})$$ for every $A,a_i$ and $\mu\in\Pi_\rho$.
\end{enumerate}

Thus the minimal requirement for an ontic model is that on average it reproduces the Born rule.
It is further standard practice to endow the probability distribution $\mu$ associated with some quantum state $\rho$ with an epistemic interpretation: it represents the ignorance regarding the true (ontic) state of the system.
It further deserves to be noted that, in general, multiple distinct physical quantities will correspond to the same POVM.
Or, in other words, a single POVM will have multiple distinct representations in the ontic model, i.e., the model will be contextual.
In addition, the set $\Pi_\rho$ will typically not be a singleton set, as preparation contextuality is also to be expected \citep{Spekkens05}.

The $\psi$-ontic/epistemic distinction was introduced by \citet{HarriganSpekkens10}.
The idea is that an ontic model is $\psi$-ontic if 
\begin{quote}
every complete physical state or ontic state in the theory is consistent with only one pure quantum state. \citep[p.126]{HarriganSpekkens10}
\end{quote}
So the ontic state determines the quantum state and, if a particular quantum state is prepared, then with certainty an ontic state obtains that is associated with that quantum state.
This coincides with the notion of quantum state representationalism if one assumes ontic states encode physical possibilities.
The further narrative is that, if a given ontic model is $\psi$-ontic, and if one interprets ontic states as encoding the properties of a system, then it reasonably follows that, in one way or another, pure quantum states correspond to possible properties of systems. 

In broad terms, this all seems somewhat reasonable.
But things get difficult when proposing a more concrete ontology.
For ontic states it is less clear how exactly they are taken to relate to the world, than it is for quantum states.
Quantum states are given within a rich mathematical structure that may be taken more or less seriously when it comes to interpreting quantum mechanics.
One popular view is to adopt a wave function representation of pure states in Hilbert space.
Quantum states then correspond to (equivalence classes of) square-integrable functions $\psi:\mathbb{R}^{3N}\to\mathbb{C}$.
It may then be proposed that this $3N$-dimensional space is to be taken seriously as a candidate for representing the true space of our world \citep{Albert96}.
Or one could insist that the ontology is to be given in terms of stuff in ``ordinary'' 3-dimensional space --a primitive ontology \citep{Goldstein98}-- and that the relation between quantum states and this stuff is only indirect.
Or the relation between quantum states and the emergent objects from daily life could be more involved still \citep{WallaceTimpson10}.
And one could go on.
In short, within quantum mechanics there is a lot of additional structure that may be taken into account when reflecting on how quantum states relate to the system that is taken to be described by it.

Ontic models do not necessarily come equipped with a specific structure that aids in the interpretation of ontic states.
And if all that is known of an ontic model is that it must be $\psi$-ontic, then this doesn't say much about how to interpret quantum states, because that will depend on what the ontic state is taken to represent.
This may be seen as a strength of $\psi$-ontology theorems. 
If successful, they show that pure quantum states represent properties of systems, while being permissive about how this representation is to work exactly.
Even a nomic reading of quantum states as preferred by some Bohmians \citep{Durr97} may be compatible with the $\psi$-ontic claim.

But this generality also comes at a price.
The interpretation of ontic states cannot be fixed by any $\psi$-ontology theorem and this also means that the interpretation of quantum states is still much up for debate even if one accepts the use of ontic models and is able to show that they must be $\psi$-ontic. 
To illustrate this point consider an extreme example.
As far as quantum Bayesianism is a viable candidate for interpreting quantum mechanics, ``ontic model Bayesianism'' may be a viable candidate for interpreting ontic models.
That is, since the main role of ontic states is determining probabilities for the outcomes of measurements, there is no principled reason to not interpret ontic states as representing degrees of belief.
If one interprets ontic states as possible expert opinions, then, if the ontic model is $\psi$-ontic, this only implies that distinct pure states correspond to poolings of expert opinions for disjoint sets of possible opinions. 
But that these sets of expert opinions are disjoint does not mean that they must necessarily attribute distinct properties to systems.\footnote{But see \citep{Myrvold20} for an argument that QBists should nevertheless adopt an ontic reading of quantum states.}

To illustrate the point further, consider the following example.
I may have complete confidence in some weather app when it comes to setting my degrees of belief about whether it will rain tomorrow.
You may have complete confidence in a different weather app that incidentally makes a different predictions for the probability of it raining tomorrow.
In this case, our epistemic states have zero overlap, but our beliefs about it raining tomorrow are not incompatible and our beliefs about the correctness of these weather apps is only incompatible if it is assumed that there is some ``true probability'' of it raining tomorrow.
But that is specifically what is denied by Quantum Bayesians who take their views about probability from subjectivists like De Finetti.

It may be noted that all these considerations about ontic states need not undermine quantum state representationalism.
It all depends on how one wishes to understand ``physical possibilities''.
In fact, QBism itself may be seen as endorsing quantum state representationalism if one takes epistemic states of rational agents to supervene on physical possibilities.
But presumably the aim of $\psi$-ontology theorems is to demonstrate something stronger.
The intended reading is that ontic states pertain in some way to the system specifically and not just to physical possibilities broadly construed.
For the remainder I shall abide by this intended reading.


\section{Ontic Models as Reducing Theories}\label{ReducSec}
In the previous section I clarified that, even if one can show for some quantum system that every ontic model must be $\psi$-ontic, this does not imply that quantum states should be given an ``ontic'' interpretation. 
But this does not mean that $\psi$-ontology theorems are devoid of content.
I agree with \citet{JenningsLeifer16} that they fall in the camp of theorems, like Bell's theorem, that chart in which sense there can be ``no return to classical reality''.
These theorems pose serious constraints on any future theory that can (approximately) reproduce the predictions of quantum mechanics in a regime where there is strong experimental support that these predictions are correct.

When bracketing the interpretation of ontic states, $\psi$-ontology theorems may be understood as aiming to establish a certain inter-theoretic relation between quantum mechanics and any future theory that is to replace it.
Showing that ontic models are necessarily $\psi$-ontic would go a long way in showing that, whatever the state space of some future theory, no ontic state can be compatible with more than one pure quantum state.
Thus there will be a bridge law between these ontic states and pure quantum states.
The bridge law is just the rule that assigns to each ontic state its corresponding quantum state.

As I shall demonstrate in Section \ref{NotShowSec}, the existence of such a rule does not trivially follow from existing $\psi$-ontology theorems.
But suppose the existence of such bridge laws can be demonstrated.
Several questions would arise.
Would they justify the claim that quantum states are part of the new theory?
Does the ontic model then really provide a lower-level explanation of what quantum states are?
And will it fully capture the theoretical role that quantum states are taken to play?

A good starting point is to look at how quantum states are introduced in ontic models.
First and foremost, quantum states are associated with probability distributions over ontic states.
This structure is reminiscent of probability distributions in statistical mechanics.
In fact, ontic models seem to be based on Einstein's idea that this is just how one should try to think about quantum mechanics in light of a possible future complete theory:
\begin{quote}
Assuming the success of efforts to accomplish a complete physical description, the statistical quantum theory would, within the framework of future physics, take an approximately analogous position to the statistical mechanics within the framework of classical mechanics. \citep{Einstein49}
\end{quote}

Such an accomplishment would be a hopeful goal if statistical mechanics were a properly understood theory without controversy.
But the justification of the use of probabilities and their meaning in statistical mechanics is by no means clear \citep{Uffink07,Frigg08}.
Trading in the controversies of the foundations of quantum mechanics for those of statistical mechanics may be a jump from the frying pan into the fire.
Nevertheless, it is worthwhile to study the analogy in some more detail.

Immediately after giving their first characterization of $\psi$-ontic models (quoted earlier), Harrigan and Spekkens provide a second characterization:
\begin{quote}
In $\psi$-ontic models, distinct quantum states correspond to disjoint probability distributions over the space of ontic states, whereas in $\psi$-epistemic models, there exist distinct quantum states that correspond to overlapping probability distributions. \citep[p.126]{HarriganSpekkens10}
\end{quote}
The idea is then that the ontic states in the support of a distribution associated with $\psi$ can be unambiguously associated with $\psi$.
As illustrated in Figure \ref{figuur}, in a $\psi$-ontic model the ontic state can tell us what the quantum state is, whereas in a $\psi$-epistemic model this is not always the case.
That this idea is not watertight is the main concern in the next section.
But for now I will go along with it.

\begin{figure}[ht]
\begin{center}
\includegraphics[width=0.8\textwidth]{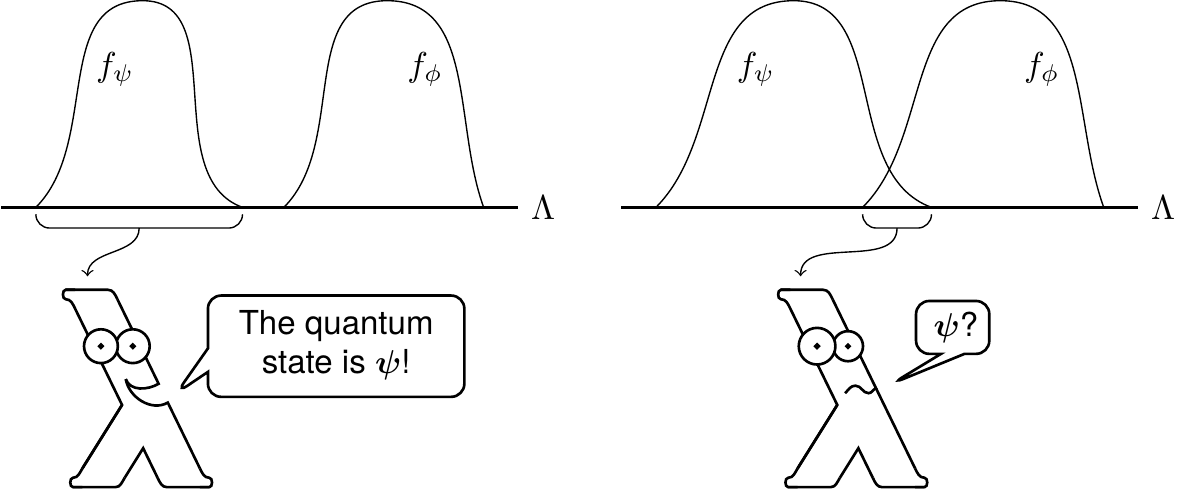}
\end{center}
\caption{Pictures clarifying the intuition behind the $\psi$-ontic/$\psi$-epistemic distinction. Here $f_\psi$ and $f_\phi$ are probability densities for probability measures associated with the quantum states $\psi$ and $\phi$. In this picture, no overlap suggests that ontic states can be associated with quantum states in an unambiguous way. When there is an overlap, the ontic states in the overlap cannot discern the quantum states $\psi$ and $\phi$.}\label{figuur}
\end{figure}

There is an interesting analog in statistical mechanics.
For the set of micro-canonical ensembles it is the case that the probability distributions are pairwise non-overlapping. 
A micro-canonical ensemble corresponds to the uniform distribution over the states with the same energy, for a fixed value of the energy.
Every micro state can thus unambiguously be associated with one micro-canonical ensemble, namely, that corresponding to the energy of that micro state.
But not many would accept this state of affairs as an argument in favor of ``micro-canonical ensemble representationalism''.
At least some qualifying remarks would have to be made.

In the case of ontic models it may be that quantum states track some property of systems in the same way that micro-canonical ensembles track the energy of systems.
But just like how Hamiltonians and micro-canonical ensembles are not the same thing, it is not to be expected that the property picked out by a probability distribution associated with a state $\psi$ is the same thing as that distribution.
After all, the support of a probability distribution (that what is taken to pick out the relevant property), has a lot less structure than the distribution itself.
Consequently, it is also not to be expected that the property picked out can in any relevant sense be identified with $\psi$ itself.
Again, this does not undermine quantum state representationalism broadly construed.
But it does highlight that quantum state representationalism should then not be taken to imply that quantum states, as understood in quantum mechanics, are a part of the physical possibilities to be associated with them.

There is a further, tangent point worth making.
That the probability distribution associated with $\psi$ might be used to pick out a property of the system does not imply that this probability distribution should be given an objective interpretation.
Regarding the role the quantum state plays as giving rise to probabilities over ontic states, it may still be given an epistemic interpretation.
Just like how micro-canonical ensembles may be given an epistemic interpretation despite picking out the energy of the system.
Moreover, the picked out property also need not correspond to the probabilities associated with quantum states by the Born rule.
If an ontic model is $\psi$-ontic, this does not imply that the probabilities from quantum theory should be interpreted as objective chances.
Thus $\psi$-ontology theorems have little to say about the meaning of probability in quantum mechanics or in future theories of physics.\footnote{There seems to have been a bit of confusion on this point, for example in \citep{BenMenahem20}.}

It may be objected at this point that there is an important dis-analogy between the probability distributions associated with quantum states in ontic models and micro-canonical ensembles in statistical mechanics.
The micro-canonical ensembles are just a special subset of all possible probability distributions and even a special subset of the stationary distributions.
When an epistemic interpretation of probability is adopted in statistical mechanics, there is in general no principled reason why, \emph{a priori}, certain distributions should play a special role when it comes to epistemic attitudes.
Additional considerations are needed to single out a certain subset of probability distributions.
But as part of the larger set of probability distributions, the micro-canonical ones do have overlaps with other distributions.
One may think that perhaps this lies behind the legitimacy of an epistemic interpretation.

I maintain that these intuitions are misguided.
Although in discussions of ontic models the probability distributions associated with quantum states play a special role, this is only because these models are investigated as possible candidate models that reproduce the quantum mechanical predictions.
But in principle there is no reason to not view these distributions as just a subset of all the admissible probability distributions.
For example, the possibility of non-quantum probability distributions is a serious topic of discussion in Bohmian mechanics \citep{Valentini20}.
And one need not even consider exotic cases.
When the distributions associated with pure quantum states are seen as just a special subset of all distributions associated with arbitrary quantum states, then there will also be overlaps.
If one holds that there is no principled conceptual distinction between mixed and pure states, then the analogy with statistical mechanics again lines up.

To sum up, in the best case scenario, $\psi$-ontology theorems provide a bridge law that associates with each ontic state a unique quantum state.
But the relationship between the ontic states that pick out a certain quantum state $\psi$ and $\psi$ as used in quantum mechanics is unclear.
If some ontic state $\lambda$ picks out the state $\psi$, this does not imply that a system in the state $\lambda$ should behave in any relevant sense as a quantum system in the state $\psi$.
So the best case scenario doesn't really justify the claim that ``quantum states are real''.
In the next section, I shall argue that the situation is actually worse than this best case scenario.


\section{What \texorpdfstring{$\boldsymbol{\psi}$}{psi}-Ontology Theorems do Not Show}\label{NotShowSec}

\subsection{\texorpdfstring{$\boldsymbol{\psi}$}{psi}-Ontic vs \texorpdfstring{$\boldsymbol{\psi}$}{psi}-Determinate Models}
When moving away from the intended interpretation of ontic models and their components, what remains of the concept of $\psi$-onticity is that of an inter-theoretic relation.
If successful, $\psi$-ontology theorems would show that in any future theory states can be linked up with quantum states.
The general strategy to try to prove this is to show that probability distributions associated with distinct pure quantum states are non-overlapping.
For example, \citet[p.477]{PBR12} state that ``[a]n important step towards the derivation of our result is the idea that the quantum state is physical if distinct quantum states correspond to non-overlapping distributions for $\lambda$.''
However, if with the quantum state being physical one means that ontic states can be unambiguously associated with unique pure quantum states, then this inference is not logically valid.
It is worthwhile to make the issue more precise.

The idea of a $\psi$-ontic model as introduced by \citet{HarriganSpekkens10} is that each ontic state is compatible with at most one pure quantum state.
I shall call an ontic model that satisfies this criterion $\psi$-determinate and reserve the term $\psi$-ontic for the related notion of non-overlapping distributions that is used in $\psi$-ontology theorems.\footnote{\Citet{Montina15} uses the term ``$\psi$-ontic in the strong sense'' instead of ``$\psi$-determinate''.}

Consider an ontic model $(\Lambda,\Sigma)$ for a fragment $(\h,\mathcal{P},\mathcal{M})$, where $\mathcal{P}$ also contains pure states.
Let $[\psi]$ denote the 1-dimensional projection on the line spanned by the vector $\psi\in\h$.
The ontic model is called \emph{$\psi$-determinate} if there exists a map that assigns to each pure state $[\psi]\in\mathcal{P}$, a measurable set $\Lambda_{[\psi]}\in\Sigma$ such that for all pairs of pure states $[\psi],[\phi]\in\mathcal{P}$ 
\begin{enumerate}[(i)]
	\item \label{PsiDet1} $\Lambda_{[\psi]}\cap\Lambda_{[\phi]}=\varnothing$, whenever $[\psi]\neq[\phi]$,
	\item \label{PsiDet2} $\mu(\Lambda_{[\psi]})=1$ for all $\mu\in\Pi_{[\psi]}$.
\end{enumerate}

Condition \ref{PsiDet1} ensures that if $\lambda\in\Lambda_{[\psi]}$, then one can unambiguously conclude that $\psi$ is the corresponding quantum state (up to a phase factor).
Condition \ref{PsiDet2} implies that if a system is prepared according to $\psi$, then with probability one an ontic state is prepared for which $\psi$ is the corresponding quantum state.
It is allowed that the union $\bigcup_{[\psi]\in\mathcal{P}}\Lambda_{[\psi]}$ is a proper subset of $\Lambda$.
So there may be ontic states that only correspond to some mixed state, but not to any pure state.
There can even be ontic states that do not correspond to any quantum state at all and thus truly go beyond quantum mechanics.
Allowing for this is inconsequential for the points to be made.

It was immediately recognized by \citet{HarriganSpekkens10} that $\psi$-determinateness implies that distinct quantum states should correspond to non-overlapping probability distributions.
It is this criterion that is used in most $\psi$-ontology theorems and formalized explicitly as the criterion of $\psi$-onticity in \citep{Leifer14}.
Specifically, an ontic model $(\Lambda,\Sigma)$ for a fragment $(\h,\mathcal{P},\mathcal{M})$ is called \emph{$\psi$-ontic} if for every distinct pair $[\psi],[\phi]\in\mathcal{P}$ the corresponding probability distributions are non-overlapping, i.e., the variational distance between their corresponding probability measures equals 1:
 \begin{equation}\label{psi-ont-def}
 	\sup_{\Delta\in\Sigma}\left|\mu(\Delta)-\nu(\Delta)\right|=1
\end{equation}
for all $\mu\in\Pi_{[\psi]},\nu\in\Pi_{[\phi]}$.
The model is called \emph{$\psi$-epistemic} if it is not $\psi$-ontic.

There is a strong intuition that any ontic model is $\psi$-determinate if and only if it is $\psi$-ontic.
Figure \ref{figuur} illustrates this.
When two distributions corresponding to distinct pure quantum states do not overlap, their supports must be disjoint. 
And so one can associate the quantum state $\psi$ with any ontic state in the support of the distribution corresponding to $\psi$.

But part of the work in this reasoning is being done by the familiar topology of the line.
In Figure \ref{figuur} the supports of the distributions form neat intervals of the line.
But using the line is just a convenient abstraction for drawing pictures.
The support of a distribution is only defined up to sets of measure zero.
In the case of Figure \ref{figuur} the natural choice is of course a single interval.
But in general there need not be a natural way to pick out one set as \emph{the} support of the distribution.

A second problem is that the choice for the support of each distribution has to be made in such a way that they are non-overlapping for \emph{all pairs} of pure quantum states.
To elaborate, consider an arbitrary $\psi$-ontic model.
Given a pair of quantum states $\psi,\phi$, it can be straightforward to construct distributions $f_\psi,f_\phi$ that have disjoint supports .
But it is not trivial to introduce a single background measure with respect to which every quantum state $\psi$ has a distribution $f_\psi$ in such a way that  $f_\psi$ and $f_\phi$ have disjoint supports for every pair of quantum states. 

It \emph{is} true that every $\psi$-determinate model is $\psi$-ontic.
This follows because the supremum in \eqref{psi-ont-def} is attained for the choice $\Delta=\Lambda_{[\psi]}$.
In other words, if $\Lambda_{[\psi]}$ is already given, it is obvious that this set should be chosen as the support for the probability distribution.
The converse, however, is not true in general.
An example of a model for a single qubit that is $\psi$-ontic but not $\psi$-determinate was given by \citet{Montina11}.

It is not clear if more general models exist that are $\psi$-ontic but not $\psi$-determinate.
In the next two subsections I consider two sets of ontic models for which it is relatively straightforward to show that they are $\psi$-ontic.
However, the question if they are also $\psi$-determinate is highly non-trivial and I have only been able to show that they are by invoking the continuum hypothesis.
Natural candidates for a map $[\psi]\mapsto\Lambda_{[\psi]}$ tend to fail and the unnatural candidates are highly non-unique.
This non-uniqueness implies that even if one counts these models as $\psi$-determinate (e.g. by accepting the continuum hypothesis), there is still no \emph{unambiguous} way to associate a quantum state with a given ontic state.
In the face of such complexity it is doubtful to conclude quantum state representationalism even if one accepts the intended interpretation of ontic states and has a $\psi$-determinate ontic model.

\subsection{Gudder's Model}\label{GudderSec}
In this section I consider a value definite ontic model for the fragment $(\h,\mathcal{P},\mathcal{M})$ with $\h$ a Hilbert space with finite dimension $d>2$, $\mathcal{P}$ the set of all pure states, and $\mathcal{M}$ the set of all self-adjoint operators.
It is a type of minimal hidden variable theory as described by \citet{Gudder70}.

The Kochen-Specker Theorem implies that it is impossible to assign unique definite values to all quantum observables in a consistent way.
Hidden variable theories can deal with this obstacle by allowing the definite value assigned to a quantum observable to depend on the context within which the observable is measured.
From an algebraic point of view, the most obvious choice for making the notion of a context precise, is by associating it with a maximal observable $A$, i.e., an observable with a non-degenerate spectrum.
Such a maximal observable defines a unique orthonormal basis of eigenstates $e_1,\ldots,e_d$ (up to phase).
Now take as the context the set of corresponding 1-dimensional projection operators.
Thus the context $C_A$ defined by $A$ is 
\begin{equation}
	C_A\ldef\{[e_1],\ldots,[e_d]\}.
\end{equation}
The possible observables that can be measured within the context $C_A$ are those whose operator commutes with $A$ or, equivalently, with all of the $[e_i]$.
The set of all contexts is denoted by $\mathfrak{C}$.

An ontic state is a rule that picks out for every context the ``true'' projection operator within that context.
Thus the set of ontic states is
\begin{equation}
	\Lambda_{\mathrm{G}}\ldef\left\{\lambda:\mathfrak{C}\to L_1(\h)\:\middle|\:\lambda(C)\in C\right\},
\end{equation}
where $L_1(\h)$ denotes the set of 1-dimensional projection operators.
Each ontic state determines contextual definite values for all observables in a straightforward way.
Let $A$ be any observable and $C$ a context such that $A$ commutes with all projections in $C$.
Then
\begin{equation}
	v_\lambda[A|C]\ldef\mathrm{Tr}(A\lambda(C))
\end{equation}
is the value of $A$ in the context $C$ when the state is $\lambda$.
It is easy to check that this ensures that $v_\lambda[A|C]$ is an element of the spectrum of $A$ and that $v_\lambda[f(A)|C]=f(v_\lambda[A|C])$ for any real-valued function $f$.
Thus within each contexts the functional relations between observables are respected.

The state space is turned into a measurable space by treating $\lambda$ as a stochastic process with $\mathfrak{C}$ playing the role of time.\footnote{It is of course much more common to use $\mathbb{R}$ or $\mathbb{Z}$ to represent time, because these sets have a natural ordering. But to the extend that I make use of the theory of stochastic processes here, the existence of such an ordering is irrelevant.}
What this means is that the $\sigma$-algebra is defined in the following way.
For any finite string of contexts $C_1,\ldots,C_n$ and quantum states $\psi_1,\ldots,\psi_n$ satisfying $[\psi_i]\in C_i$ define
\begin{equation}\label{cylset}
	\Delta_{C_1,\ldots,C_n}^{\psi_1,\ldots,\psi_n}\ldef\left\{\lambda\in\Lambda_{\mathrm{G}}\:\middle|\:\lambda(C_i)=[\psi_i]~\forall i\right\}
\end{equation}
the set of ontic states that pick out $[\psi_i]$ as the true projection in the context $C_i$ for every $i$.
Then take $\Sigma_{\mathrm{G}}$ to be the smallest $\sigma$-algebra containing all sets of the form \eqref{cylset}.
A probability measure on $(\Lambda_{\mathrm{G}},\Sigma_{\mathrm{G}})$ is completely determined by its action on the sets of the form \eqref{cylset}.\footnote{This follows from Kolmogorov's Extension Theorem.}

For a pure state $[\psi]\in\mathcal{P}$, the corresponding probability measure $\mu_{[\psi]}$ is defined as
\begin{equation}\label{probmeas}
	\mu_{[\psi]}\left(\Delta_{C_1,\ldots,C_n}^{\psi_1,\ldots,\psi_n}\right)\ldef\prod_{i=1}^n\left|\left\langle\psi\middle|\psi_i\right\rangle\right|^2.
\end{equation}
Within each context, $\mu_{[\psi]}$ reproduces the Born rule for all observables in that context.
On the other hand, observables in distinct contexts are treated as stochastically independent, even if they are associated with the same self-adjoint operator.

It is far from trivial to try to think of this model as a model in which ``quantum states are real''.
The main role of any ontic state is to assign (contextual) definite values to all observables.
But such definite values have little to say about quantum states.
And so, for a given ontic state, it is not clear which, if any, quantum state to associate with it.
Nevertheless, it is not very difficult to show that the model is in fact $\psi$-ontic.

Consider an arbitrary pure state $[\psi]\in\mathcal{P}$ and let $(C_i^\psi)_{i=1}^\infty$ be some countable sequence of distinct contexts such that $[\psi]\in C_i^\psi$ for every $i$.\footnote{This is the point where it is relevant that $d>2$ because if $d=2$ every $[\psi]$ belongs to only one context.}
Now define
\begin{equation}\label{psi-set}
	\Lambda_{[\psi]}\ldef\bigcap_{i=1}^\infty\Delta_{C_i^\psi}^{\psi}.
\end{equation}
This is a measurable set of ontic states that satisfies
\begin{equation}
	\mu_{[\phi]}\left(\Lambda_{[\psi]}\right)=
	\prod_{i=1}^\infty
	\left|\left\langle\psi\middle|\phi\right\rangle\right|^2=
	\begin{cases}
		1&[\psi]=[\phi],\\
		0&\text{otherwise}
	\end{cases}
\end{equation}
for every $[\phi]$.
Thus for any pair of distinct states $[\psi],[\phi]$ the corresponding probability measures $\mu_{[\psi]},\mu_{[\phi]}$ are non-overlapping.

The given construction for associating measurable sets of ontic states with pure quantum states \eqref{psi-set} does not suffice to show that the model is $\psi$-determinate.
For this it would be required that each ontic state is associated with at most one quantum state: condition \ref{PsiDet1} should be satisfied.
But whenever $[\psi]$ and $[\phi]$ do not commute, if $\lambda\in\Lambda_{[\psi]}$, this poses no constraints on the values $\lambda$ can take on for the contexts $C_i^\phi$.
Thus $\Lambda_{[\psi]}\cap\Lambda_{[\phi]}$ is not empty and knowing that $\lambda\in\Lambda_{[\psi]}$ does not allow one to infer that $[\psi]$ is the quantum state of the system.

A natural first reaction to $\Lambda_{[\psi]}$ as defined by \eqref{psi-set}, is that it is simply too big.
Thinking about what it means for an ontic state $\lambda$ to be a $[\psi]$-state, one may start with the requirement that $\lambda$ at least picks out the state $[\psi]$ in any context where this is possible.
That is, $\lambda$ should lie in the set
\begin{equation}
	\Lambda_{[\psi]}^\circ\ldef\bigcap_{\mathclap{\left\{\substack{C\in\mathfrak{C};\\ [\psi]\in C}\right\}}}\Delta_C^{\psi}
	=\left\{\lambda\in\Lambda_{\mathrm{G}}\:\middle|\:[\psi]\in C\implies\lambda(C)=[\psi]\right\}.
\end{equation} 

There are several problems when heading in this direction.
First, it is still the case that $\Lambda_{[\psi]}^\circ\cap\Lambda_{[\phi]}^\circ$ is non-empty whenever $[\psi]$ and $[\phi]$ do not commute.
But more importantly, the set $\Lambda_{[\psi]}^\circ$ is not $\Sigma_{\mathrm{G}}$-measurable (proven below), and so it cannot be used to show that $(\Lambda_{\mathrm{G}},\Sigma_{\mathrm{G}})$ is a $\psi$-determinate model.
One might object that this is just a sign that also $\Sigma_{\mathrm{G}}$ is too small.
There may be the intuition to insist that sets like $\Lambda_{[\psi]}^\circ$ should just be measurable.
But this intuition seems to rely on taking quantum states as a natural ingredient to constructing ontic models, which would be question begging.
The fact is that, from the point of view of the ontic model, there is no obvious reason to include $\Lambda_{[\psi]}^\circ$ as a measurable set.
But even if one does have a non-question begging argument for including the set, it is not helpful as the following theorem demonstrates.\footnote{A proof for this theorem is given in Appendix \ref{Proof1App}.}

\begin{theorem}\label{examplethm}
The set $\Lambda_{[\psi]}^\circ$ is not $\Sigma_{\mathrm{G}}$-measurable and there exist extensions $\mu_{[\psi]}^+$ and $\mu_{[\psi]}^-$ of $\mu_{[\psi]}$ to the smallest $\sigma$-algebra containing both $\Sigma_{\mathrm{G}}$ and $\Lambda_{[\psi]}^\circ$ such that
\begin{equation}
	\mu_{[\psi]}^+\left(\Lambda_{[\psi]}^\circ\right)=1~\text{and}~\mu_{[\psi]}^-\left(\Lambda_{[\psi]}^\circ\right)=0.
\end{equation}
\end{theorem}

This theorem shows that even if one includes a set like $\Lambda_{[\psi]}^\circ$, this does not mean that the set will have probability one according to the quantum state $[\psi]$.
It may even have probability zero.
The upshot is that the natural step closer to satisfying condition \ref{PsiDet1} of $\psi$-determinateness takes us a step further away from satisfying condition \ref{PsiDet2}.
In a sense, the set $\Lambda_{[\psi]}^\circ$ is simply too small.

A different strategy is needed to determine if the model is $\psi$-determinate.
This requires a better grasp on what kind of maps between pure quantum states and measurable subset of $\Lambda_{\mathrm{G}}$ are possible.
Consider again the map given by (\ref{psi-set}).

Whatever the set of ontic states corresponding to $[\psi]$ is going to be, it might as well be a subset of this set since it has $\mu_{[\psi]}$-probability 1.
Now consider any pair of distinct pure states $[\psi],[\phi]$.
The easiest way to modify $\Lambda_{[\psi]}$ and $\Lambda_{[\phi]}$ to ensure that their intersection is empty, is to intersect one with the complement of the other.

For any $[\psi]$ one can intersect $\Lambda_{[\psi]}$ with \emph{countably} many sets $\Lambda_{[\phi]}^c$ and the resulting set is still measurable and has $\mu_{[\psi]}$-probability 1.
This is clearly a step in the right direction.
But there are \emph{uncountably} many $\Lambda_{[\phi]}$ that have a non-empty intersection with $\Lambda_{[\psi]}$ and taking the intersection with \emph{all} sets $\Lambda_{[\phi]}^c$ again results in a non-measurable set.
The problem may be circumvented if for any $[\phi]$ for which $\Lambda_{[\psi]}$ is not modified by intersecting it with $\Lambda_{[\phi]}^c$, instead, $\Lambda_{[\phi]}$ is modified by intersecting it with $\Lambda_{[\psi]}^c$.
Thus for any pair $[\psi],[\phi]$ one needs to either intersect $\Lambda_{[\psi]}$ with $\Lambda_{[\phi]}^c$, or $\Lambda_{[\phi]}$ with $\Lambda_{[\psi]}^c$ such that for any $[\psi]$ the set $\Lambda_{[\psi]}$ is being intersected with only countably many $\Lambda_{[\phi]}^c$.
This turns out to be possible if and only if the continuum hypothesis holds.

To make this strategy more precise, consider maps $m:L_1(\h)\to L_1(\h)^\mathbb{N}$ that satisfy 
\begin{equation}
	m^{[\psi]}_n\ldef(m([\psi]))(n)\neq[\psi]~
	\forall n\in\mathbb{N},[\psi]\in L_1(\h).
\end{equation}
For any such $m$ call 
\begin{equation}
	\Lambda^m_{[\psi]}\ldef\Lambda_{[\psi]}\cap\bigcap_{n\in\mathbb{N}}\Lambda^c_{m^{[\psi]}_n}
\end{equation}
a \emph{canonically modified $\psi$-ontic subset}.
The following theorem then holds.\footnote{A proof is given in Appendix \ref{Proof2App}.}

\begin{theorem}\label{mainthm}
There exists a map $m$ such that the canonically modified $\psi$-ontic subsets make the ontic model $(\Lambda_G,\Sigma_G)$ $\psi$-determinate if and only if the continuum hypothesis holds.
\end{theorem}

Canonically modifying the $\psi$-ontic subsets is not the only possible strategy for trying to make the model $\psi$-determinate.\footnote{Here is another possible strategy. Suppose $[\psi],[\phi]$ are given and non-commuting. Let $\phi'$ be any state such that $[\phi'][\phi]=0$ and $[\phi'][\psi]\neq0$ and let $(C_i^{\phi'})_{i=1}^\infty$ be a countable sequence of distinct contexts containing $[\phi']$. Each set $\Delta_{C_i}^{\phi'}$ has $\mu_{[\phi]}$-probability zero, but a positive probability according to $\mu_{[\psi]}$. The countable union $\bigcup_{i=1}^\infty\Delta_{C_i}^{\phi'}$ therefore also has $\mu_{[\phi]}$-probability zero, but $\mu_{[\psi]}$-probability one. So one can modify $\Lambda_{[\psi]}$ and $\Lambda_{[\phi]}$ by intersecting the first with $\bigcup_{i=1}^\infty\Delta_{C_i}^{\phi'}$ and the second with the complement thereof.} 
Therefore it is not clear if the continuum hypothesis is necessary for showing that the ontic model is $\psi$-determinate.
But it may well be the case that whether the ontic model $(\Lambda_{\mathrm{G}},\Sigma_{\mathrm{G}})$ is $\psi$-determinate or not depends on whether one accepts the continuum hypothesis.
If so, I don't think the main lesson would be about the validity of the continuum hypothesis.\footnote{Although it is interesting to note that the version of the continuum hypothesis used here is precisely the one used by \citet{Freiling86} to argue against its validity.}
What is important in that case, is that it is impossible to give an explicit map $[\psi]\mapsto\Lambda_{[\psi]}$ that makes the model $\psi$-determinate.
Without an explicit map, there is no meaningful sense in which a quantum state can be considered to be part of the ontic state of a system in this model, \emph{even though the model is $\psi$-determinate.}
Before discussing the further lessons from this model regardless of whether the continuum hypothesis is needed, I first consider a related set of ontic models.

\subsection{The MKC Models}\label{MKCSection}
The ontic models due to \citet{Meyer99}, \citet{Kent99} and \citet{CliftonKent00} were specifically designed to show the possibility of noncontextual value definite ontic models that can reproduce the predictions of quantum mechanics with arbitrary precision.
Thus, even though the Kochen-Specker Theorem rules out noncontextual models that reproduce all of quantum mechanics, there are noncontextual models that are empirically indiscernible from quantum mechanics \citep{BarrettKent04,Hermens11}.
They are of interest here, because they provide even less handles to associate quantum states with ontic states than the models by \citet{Gudder70}.
Nevertheless, it is again possible to show that the models are $\psi$-determinate with the use of the continuum hypothesis.

The construction of the models resembles the construction above.
Ontic states assign definite values to observables within a given context.
But instead of considering the complete set of contexts $\mathfrak{C}$, one now considers a specific \emph{countable} subset $\mathfrak{C}_{\mathrm{MKC}}\subset\mathfrak{C}$.
This set of contexts is chosen to satisfy the following two conditions:
\begin{enumerate}
\item Any context can be approximated up to arbitrary precision. For every $C=\{[e_1],\ldots,[e_d]\}\in\mathfrak{C}$ and $\epsilon>0$ there exists a $C^\epsilon=\{[e_1'],\ldots,[e_d']\}\in\mathfrak{C}_{\mathrm{MKC}}$ such that $\mathrm{Tr}([e_i][e_i'])>1-\epsilon$ for all $i$.
\item Every pair of contexts is totally incompatible. For any pair $C_1,C_2\in\mathfrak{C}_{\mathrm{MKC}}$ with $C_1\neq C_2$ and every $[\psi]\in C_1,[\phi]\in C_2$, $[\psi]$ and $[\phi]$ do not commute.
\end{enumerate}
The set of ontic states is now given by
\begin{equation}
	\Lambda_{\mathrm{MKC}}\ldef\left\{\lambda:\mathfrak{C}_{\mathrm{MKC}}\to L_1(\h)\:\middle|\:\lambda(C)\in C\right\}.
\end{equation}
Because of condition 2, every (non-trivial) observable is compatible with at most one context.
Therefore, a definite value assigned to it is automatically noncontextual.
Because of the first condition, for every observable $A$ in quantum mechanics, there is a context $C$ in the model and an observable $A'$ that is compatible with $C$, such that $A'$ is a good approximation of $A$.

The fragment of quantum mechanics $(\h,\mathcal{P},\mathcal{M})$ is now such that $\mathcal{P}$ still contains all the pure states, but $\mathcal{M}$ now only contains the self-adjoint operators that have an orthonormal basis of eigenstates that forms a context in $\mathfrak{C}_{\mathrm{MKC}}$.
The $\sigma$-algebra $\Sigma_{\mathrm{MKC}}$ is again introduced with the use of the sets from \eqref{cylset}, but now with $\Lambda_{\mathrm{G}}$ replaced by $\Lambda_{\mathrm{MKC}}$.
For any pure state $[\psi]\in\mathcal{P}$ the corresponding probability measure is again given by \eqref{probmeas}.

Showing that the MKC models are $\psi$-ontic is not entirely trivial.
For most given states $[\psi],[\phi]$, there are no contexts that contain $[\psi]$ or $[\phi]$ or a pure state perpendicular to one of these.
This is because only countably many pure states occur in some context.
Thus, for almost every pure state $[\psi]$ one has
\begin{equation}
	0<\mu_{[\psi]}(\Delta_{C_1,\ldots,C_n}^{\psi_1,\ldots,\psi_n})<1
\end{equation}
for every set of the form $\Delta_{C_1,\ldots,C_n}^{\psi_1,\ldots,\psi_n}$.

To show that the MKC models are $\psi$-ontic nonetheless, we look at the limit $n\to\infty$.
For most sequences $((C_i,[\psi_i]))_{i=1}^\infty$, the probability $\prod_{i=1}^n\left|\left\langle\psi\middle|\psi_i\right\rangle\right|^2$ will go to zero as $n\to\infty$.
But if the sequence is appropriately chosen, it will be finite.
This is applied in the proof of the following theorem (see Appendix \ref{Proof3App}).

\begin{theorem}\label{MKCTheorem}
For every pure state $[\psi]$ and every non-zero $n\in\mathbb{N}$, there exists a set $\Lambda_{[\psi],n}\in\Sigma_{\mathrm{MKC}}$ such that
\begin{equation}
	\mu_{[\psi]}\left(\Lambda_{[\psi],n}\right)>1-\frac{1}{n}~\text{and}~\mu_{[\phi]}\left(\Lambda_{[\psi],n}\right)=0
\end{equation}
for all $[\phi]\neq[\psi]$.
Thus the MKC models are $\psi$-ontic.
\end{theorem}

What is interesting about this theorem, is that it establishes that the models are $\psi$-ontic without providing a candidate for a set that is big enough to contain all the ontic states corresponding to some quantum state $[\psi]$.
But as $n$ increases, $\Lambda_{[\psi],n}$ gets closer to being such a candidate.
So as a first step to getting such a set, one can take the union of all the $\Lambda_{[\psi],n}$ for all values of $n$. 
This in turn is much larger than required, for only when $n$ approaches infinity do the ontic states in $\Lambda_{[\psi],n}$ start to ``center around'' $[\psi]$.
So instead define
\begin{equation}
	\Lambda_{[\psi]}\ldef\bigcap_{m=1}^\infty\bigcup_{n=m}^\infty\Lambda_{[\psi],n}.
\end{equation}
This set does satisfy
\begin{equation}
	\mu_{[\phi]}\left(\Lambda_{[\psi]}\right)=\begin{cases}
	1&[\phi]=[\psi],\\
	0&[\phi]\neq[\psi].
	\end{cases}
\end{equation}

Like with Gudder's model, the MKC models are not obviously $\psi$-determinate.
Whenever $[\psi]$ and $[\phi]$ do not commute, $\Lambda_{[\psi]}\cap\Lambda_{[\phi]}$ will be non-empty.
But as is the case with Gudder's model, a canonical modification of these sets can be used to make the model $\psi$-determinate if one assumes the continuum hypothesis. 

There is however also an interesting distinction.
For an ontic state $\lambda$ to correspond to the quantum state $[\psi]$ it may be taken to be a necessary condition that $\lambda\in\Lambda_{[\psi]}$.
This set is what is known as a tail event.
This means that which values $\lambda$ takes on at any finite set of contexts is irrelevant to determine if $\lambda\in\Lambda_{[\psi]}$ or not; it all hinges on how $\lambda$ behaves in the limit (the tail).
But for all practical purposes, if the model is to be used for making predictions, it suffices to know how $\lambda$ behaves for a large, but finite, set of contexts.
So which quantum state corresponds to some ontic state is more or less irrelevant from the point of view of the model.
It then seems peculiar to attach any ontological significance to this quantum state.

This point can even be exploited a bit further.
As \citet[p.159]{BarrettKent04} have discussed, when it comes to reproducing the predictions of quantum mechanics up to a finite precision, one does not need the full set of contexts $\mathfrak{C}_{\mathrm{MKC}}$.
An appropriately chosen finite subset will do as well.
Although such a model is in principle discernible from quantum mechanics, there also always exists such a model that is not discernible given the present day finite precision of measurement.
So within the bounds of such finite precision, the finite model is just as good as the full model.
But such finite models are also trivially not $\psi$-determinate; they only have a finite number of ontic states, while there are uncountably many pure quantum states.
Dragging in finite ontic models may not be an entirely fair move to make here.
But it merely serves to drive the point home that, from the perspective of the ontic states in the MKC models, what the ``true'' quantum state is does not matter.

\subsection{Conclusion}

The general strategy of $\psi$-ontology theorems is to show that distinct quantum states correspond to non-overlapping probability distributions.
But ``non-overlap'' is a somewhat misleading term here.
It suggests that the probability distributions corresponding to distinct quantum states have disjoint supports, and that quantum states can unambiguously be associated with the ontic states in those supports.
In other words, it is taken for granted that showing that a model is $\psi$-ontic is sufficient to show that it is $\psi$-determinate.
The ontic models just discussed demonstrate that this is not an innocent inference.

As demonstrated, if one assumes that the continuum hypothesis is true, then these models are $\psi$-determinate.
It is not clear if the continuum hypothesis is in fact necessary.
If it is, then we have here an example of $\psi$-determinate ontic models for which it is questionable that they endorse quantum state representationalism, even in the broad sense.
It is true that in this case there exists a map that makes it the case that changing the quantum state necessitates a change in the ontic state. 
But since this map only exists in the Platonic realm as some inaccessible mathematical entity, there is no way of telling what kind of change in the ontic state is required.
From the point of view of the, allegedly more fundamental, ontic state, there is also no way of telling what the quantum state is.
And whatever the quantum state is, it is irrelevant from the point of view of the ontic state.
So there is no merit in granting the quantum state any metaphysical status here.

Would the case for quantum state representationalism be helped if the continuum hypothesis turned out to not be necessary?
I don't think so.
In this case, it would be possible for any of the models to give an explicit map $r:\mathcal{P}\to\Sigma$, $r:[\psi]\mapsto\Lambda_{[\psi]}^r$ that makes the model $\psi$-determinate.
But this map would also be non-unique.
So if $\lambda\in\Lambda_{[\psi]}^r$ for some such map, but also $\lambda\notin\Lambda_{[\psi]}^{r'}$ for some other map $r'$, how would we decide if $[\psi]$ is the true quantum state or not if the ontic state is $\lambda$?

The problem is that one can shuffle around sets of ontic states of measure zero from one quantum state to another to modify a given map $r$.
Of course, shuffling around such sets is always possible in probability theory.
In many cases this is unproblematic because there is a natural candidate for how to do this like in Figure \ref{figuur}.
But in the case of the ontic models just discussed there is no natural way to remove the ambiguity.

The reason behind this is that the ontic states in these models do not care about quantum states.
What matters from the point of view of the ontic model, is which observables have which values. 
These may be considered the genuine properties of the system. 
If by some rule, some attribution of definite values also gives rise to some quantum state, this information is irrelevant and arbitrary.

This point is even stronger in the case of the MKC models.
In these models most quantum states\footnote{All except a countable set.} have to be associated with so-called tail events.
This means that for almost every quantum state $[\psi]$, any specification of a finite number of definite values of physical quantities is completely irrelevant for determining whether an ontic state $\lambda$ is to correspond to $[\psi]$ or not.
For a given $\lambda$, only its behavior ``in the limit'' determines to which quantum state it corresponds.

Another way to understand the issues at hand is by looking at an analogy with classical mechanics.
For a classical system, any physical quantity can be represented by a function on phase space.
The state of the system in this way naturally determines the value of that quantity.
But this does not mean that it is meaningful to interpret every function on phase space as a physical quantity.
Especially when that function becomes rather complex or even requires the validity of the continuum hypothesis to warrant its existence.
That would be a case of taking your mathematics too seriously for interpreting your physical model.


\section{What \texorpdfstring{$\boldsymbol{\psi}$}{psi}-Ontology Theorems do Show}\label{DoShowSec}

I want to end this paper with a more positive message.
If the $\psi$-ontic/epistemic distinction is supposed to capture whether the quantum state corresponds to an objective property of an individual quantum system or not, then the formal definition adopted in $\psi$-ontology theorems is not appropriate.
After all, these theorems do not settle the question whether ontic models should be $\psi$-determinate or not.
This does not imply that $\psi$-ontology theorems are void of content.
But what lessons can be drawn from these theorems then?
In this section I shall argue that these theorems are best seen as putting constraints on possible explanations for certain quantum phenomena.

To show that ontic models are $\psi$-ontic, additional assumptions are adopted.
The PBR Theorem for example, makes use of the assumption of preparation independence (a weak version of separability).
Over the years, it has been shown that weaker versions of this assumption also suffice \citep{Hall11,SchlosshauerFine12,Myrvold18} (see also \citep[\S7.4]{Leifer14}).
The $\psi$-ontology theorem due to \citet{ColbeckRenner17} makes use of parameter independence.
In addition, all theorems make use of a no-conspiracy assumption: which measurements are performed is independent of how the system is prepared.

The no-conspiracy assumption is part and parcel of ontic models \citep{Hermens19}.
Although rejecting it may be a reasonable way to escape the constraints posed by $\psi$-ontology theorems \citep[\S3.3]{FriederichEvans19}, I will not consider this option here.
What is interesting, is that not long after \citet{PBR12} presented their theorem, it was shown that additional assumptions on top of the no-conspiracy assumption are necessary.
Without additional assumptions, it is always possible to construct a $\psi$-epistemic model \citep{Lewis12,Aaronson13,Mansfield16}.

The overlap in probability distributions allowed in these $\psi$-epistemic models is quite small though.
It was shown by \citet{Maroney12} that this is necessarily so: even though $\psi$-epistemic models can always be constructed, they cannot be ``maximally $\psi$-epistemic''.
As a consequence, these models cannot do all of the explanatory work one would hope $\psi$-epistemic models to do.
Specifically, it has now been shown thoroughly, both theoretically and experimentally, that no such model can explain the indistinguishability of non-orthogonal quantum states \citep{Barrett14,Branciard14,Leifer14b,Ringbauer15,Nigg16,Knee17}. 
Thus even if in some ontic model a non-negligible set of ontic states is compatible with multiple quantum states, this cannot be used to explain why the success rate for finding out which quantum state was prepared with using a single shot measurement is as low as it is according to quantum mechanics.
The results in this paper do nothing to counter those conclusions.
The new claim I have made here, is that even in the limit of zero overlap, one cannot conclude that ontic states should determine quantum states.
Therefore, the conclusions that may be drawn from $\psi$-ontology theorems do not go beyond the conclusions that may be drawn from the BCLM Theorem \citep{Barrett14}, which places constraints on the possible overlaps.

It deserves to be mentioned that, from an experimental point of view, this conclusion is not entirely novel.
For example, the proof of the PBR Theorem makes use of the notion of anti-distinguishability.
A measurement with possible outcomes $a_1,\ldots,a_n$ is said to anti-distinguish the set of quantum states $\psi_1,\ldots,\psi_n$ if for every $i=1,\ldots,n$ the outcome $a_i$ has probability zero whenever the state $\psi_i$ is prepared.
However, experimentally it is impossible to verify if a certain measurement is anti-distinguishing.
Because of this, it is always possible to construct $\psi$-epistemic models that are compatible with the outcomes of experimental tests of the PBR Theorem.
More generally, the notion of non-overlapping probability distributions itself is not robust with respect to measurement errors.
Any $\psi$-ontic model may be approximated by a $\psi$-epistemic model by introducing an arbitrarily small overlap in the probability distributions.
But although the experimental tests cannot rule out $\psi$-epistemic models, they \emph{do} provide upper bounds on the overlaps of the distributions.\footnote{For further discussion see \citep{Ringbauer17}.}

To conclude, proofs for $\psi$-ontology theorems do not show the necessity of quantum state representationalism, even if one accepts all the usual premises.
Rather than forcing a certain interpretation on possible future theories, they pose constraints on the kind of explanations such theories may provide for quantum phenomena.
These kinds of explanations are of the kind offered in Spekkens' $\psi$-epistemic toy model \citep{Spekkens07}.
And to be fair to \citet{PBR12}, their theorem seems to have been designed first and foremost to demonstrate the impossibility of extending this toy model to more serious quantum systems.
The grand metaphysical claims that followed after that may have been not much more than a (successful) marketing hype.

\section*{Acknowledgments}

I would like to thank Owen Maroney for helpful discussions about $\psi$-ontology theorems and the MKC models. 
I also want to thank Harvey Brown for stimulating discussions on the reality of the quantum state, for inviting me to give a talk about this paper in Oxford, and for complimenting me on the drawings in Figure \ref{figuur}.
I further want to thank the people at the Mathematics StackExchange for leading me towards the realization that I have wasted quite a bit of time trying to disprove the continuum hypothesis.  
This research was funded by the Netherlands Organisation for Scientific Research (NWO), Veni Project No. 275-20-070.

\printbibliography

\begin{appendix}

\section{Proofs of Theorems}\label{ProofSection}

The first two theorems require some insight in the kind of sets that are and are not contained in $\Sigma_{\mathrm{G}}$.
For this we make use of the $\sigma$-algebra $\Sigma_{\mathrm{cg}}$ of countably generated subsets of $\Lambda_{\mathrm{G}}$.

\begin{definition}
For an arbitrary countable set of contexts $\mathfrak{C}_c\subset\mathfrak{C}$ define
\begin{equation}
\begin{gathered}
	\Lambda_{\mathfrak{C}_c}\ldef\left\{\lambda:\mathfrak{C}_c\to L_1(\h)\:\middle|\:\lambda(C)\in C\right\},\\
	\Pi_{\mathfrak{C}_c}:\Lambda_{\mathrm{G}}\to\Lambda_{\mathfrak{C}_c},~
	\left[\Pi_{\mathfrak{C}_c}(\lambda)\right](C)\ldef\lambda(C).
\end{gathered}
\end{equation}
A subset $\Delta\subset\Lambda_{\mathrm{G}}$ is called \emph{countably generated} if there exists a countable set $\mathfrak{C}_c$ and a subset $\Delta_{\mathfrak{C}_c}\subset\Lambda_{\mathfrak{C}_c}$ such that
\begin{equation}
	\Delta=\Pi_{\mathfrak{C}_c}^{-1}\left(\Delta_{\mathfrak{C}_c}\right).
\end{equation}
\end{definition}

Thus a subset of $\Delta\subset\Lambda_{\mathrm{G}}$ is countably generated iff for any $\lambda$ one only needs to check its action on countably many contexts to determine whether $\lambda\in\Delta$.
The following lemma has a straightforward proof which is left to the reader.
\begin{lemma}\label{cglemma}
The set $\Sigma_{\mathrm{cg}}$ of all countably generated subsets of $\Lambda_{\mathrm{G}}$ is a $\sigma$-algebra and has $\Sigma_{\mathrm{G}}$ as a sub-$\sigma$-algebra.
\end{lemma}

\subsection{Proof of Theorem \ref{examplethm}}\label{Proof1App}

The proof of Theorem \ref{examplethm} requires the following lemma:
\begin{lemma}\label{extlemma}
Let $(\Lambda,\Sigma,\mu)$ be a probability space and $A$ be a possibly non-measurable subset of $\Lambda$ such that
\begin{equation}\label{Rcond2}
	\forall\Delta\in\Sigma:~A\cap\Delta=\varnothing\implies\mu(\Delta)=0,
\end{equation}
then there exists an extension $\mu^+$ of $\mu$ to $\Sigma_A$, the $\sigma$-algebra generated by $\Sigma$ and $A$, such that $\mu^+(A)=1$.
\end{lemma}

\begin{proof}
The smallest $\sigma$-algebra containing $\Sigma$ and $A$ is given by
\begin{equation}
	\Sigma_A\ldef\left\{(A\cap\Delta_1)\cup(A^c\cap\Delta_2)\:\middle|\:\Delta_1,\Delta_2\in\Sigma\right\}.
\end{equation}
Now define $\mu^+:\Sigma_A\to[0,1]$ by
\begin{equation}
	\mu^+\left(\vphantom{\mu^+}(A\cap\Delta_1)\cup(A^c\cap\Delta_2)\right)\ldef\mu(\Delta_1).
\end{equation}
To see that this is well-defined, suppose $A\cap\Delta_1=A\cap\tilde{\Delta}_1$ for some $\Delta_1,\tilde{\Delta}_1\in\Sigma$.
Then
\begin{equation}
\begin{split}
	A\cap(\Delta_1\cap\tilde{\Delta}_1^c)=(A\cap\Delta_1)\cap\tilde{\Delta}_1^c=&A\cap\tilde{\Delta}_1\cap\tilde{\Delta}_1^c=\varnothing,\\
	A\cap(\Delta_1^c\cap\tilde{\Delta}_1)=(A\cap\tilde{\Delta}_1)\cap\Delta_1^c=&A\cap\Delta_1\cap\Delta_1^c=\varnothing.
\end{split}
\end{equation}
From \eqref{Rcond2} it follows that $\mu(\Delta_1\cap\tilde{\Delta}_1^c)=\mu(\Delta_1^c\cap\tilde{\Delta}_1)=0$ and therefore
\begin{equation}
\begin{split}
	\mu(\Delta_1)
	&=
	\mu(\Delta_1\cap\tilde{\Delta}_1)+\mu(\Delta_1\cap\tilde{\Delta}_1^c)\\
	&=
	\mu(\Delta_1\cap\tilde{\Delta}_1)+\mu(\Delta_1^c\cap\tilde{\Delta}_1)=\mu(\tilde{\Delta}_1).
\end{split}
\end{equation}
That $\mu^+$ is a probability measure follows from the fact that $\mu$ is a probability measure.
Finally, it holds that
\begin{equation}
	\mu^+(A)=\mu^+\left(\vphantom{\mu^+}(A\cap\Lambda)\cup(A^c\cap\varnothing)\right)=\mu(\Lambda)=1.
\end{equation}
\end{proof}

We now turn to the proof of Theorem \ref{examplethm}.
The set $\Lambda_{[\psi]}^\circ$ is not countably generated so by Lemma \ref{cglemma} it is not $\Sigma_{\mathrm{G}}$-measurable.
For the remainder of the proof it only needs to be shown that both $\Lambda_{[\psi]}^\circ$ and $\Lambda_{\mathrm{G}}\backslash\Lambda_{[\psi]}^\circ$ satisfy condition \eqref{Rcond2} and then Lemma \ref{extlemma} can be applied.

To see that $\Lambda_{[\psi]}^\circ$ satisfies \eqref{Rcond2} let $\Delta\in\Sigma$ be any subset that satisfies $\Delta\cap\Lambda_{[\psi]}^\circ=\varnothing$.
By Lemma \ref{cglemma} there exists a countable set $\mathfrak{C}_c$ and a subset $\Delta_{\mathfrak{C}_c}\subset\Lambda_{\mathfrak{C}_c}$ such that $\Delta=\Pi_{\mathfrak{C}_c}^{-1}\left(\Delta_{\mathfrak{C}_c}\right)$.
Now suppose $\lambda\in\Delta$.
Because $\lambda\notin\Lambda_{[\psi]}^\circ$, there exists a context $C\in\mathfrak{C}_c$ such that $[\psi]\in C$ and $\lambda(C)\neq [\psi]$.
Thus $\lambda\in\Lambda_{\mathrm{G}}\backslash\Delta_C^{\psi}$ and, more generally,
\begin{equation}
	\Delta\subset\overline{\Delta}\ldef\bigcup_{\mathclap{\substack{C\in\mathfrak{C}_c\\ [\psi]\in C}}}\Lambda_{\mathrm{G}}\backslash\Delta_C^{\psi}. 
\end{equation}
Because $\overline{\Delta}$ is the countable union of sets with $\mu_{[\psi]}$-probability zero, $\overline{\Delta}$ is itself a measurable set with $\mu_{[\psi]}$-probability zero.
Hence $\mu_{[\psi]}(\Delta)=0$ and there exists an extension $\mu_{[\psi]}^+$ such that $\mu_{[\psi]}^+(\Lambda_{[\psi]}^\circ)=1$. 

To see that $\Lambda_{\mathrm{G}}\backslash\Lambda_{[\psi]}^\circ$ satisfies \eqref{Rcond2} let $\Delta\in\Sigma$ be any subset that satisfies $\Delta\cap\Lambda_{\mathrm{G}}\backslash\Lambda_{[\psi]}^\circ=\varnothing$.
It will be shown that this only holds if $\Delta$ is empty.
Suppose towards a contradiction that $\Delta$ is not empty and $\lambda\in\Delta$.
Let $\mathfrak{C}_c$ and $\Delta_{\mathfrak{C}_c}$ again be as in Lemma \ref{cglemma}.
Because $\mathfrak{C}_c$ is countable there is a $C$ that contains $[\psi]$ such that $C\notin\mathfrak{C}_c$.
If $\lambda(C)\neq [\psi]$ then $\lambda\in\Lambda_{\mathrm{G}}\backslash\Lambda_{[\psi]}^\circ$ and a contradiction is obtained.
If $\lambda(C)=[\psi]$, then define $\lambda'$ to be equal to $\lambda$ for all contexts except for $C$ where it attains any value other than $[\psi]$.
Then $\lambda'\in\Lambda_{\mathrm{G}}\backslash\Lambda_{[\psi]}^\circ$ but also, since $C\notin\mathfrak{C}_c$, $\lambda'\in\Delta$, again giving a contradiction.
Hence $\Delta\cap\Lambda_{\mathrm{G}}\backslash\Lambda_{[\psi]}^\circ=\varnothing$ if and only if $\Delta=\varnothing$.
By Lemma \ref{cglemma} then there exists a measure $\mu^-_{\psi}$ which is an extension of $\mu_{[\psi]}$ and which satisfies  $\mu_{[\psi]}^-(\Lambda_{\mathrm{G}}\backslash\Lambda_{[\psi]}^\circ)=1$.

\subsection{Proof of Theorem \ref{mainthm}}\label{Proof2App}

Here we require the following lemma.
\begin{lemma}\label{CHlemma}
Let $X$ be a set with the cardinality of $\mathbb{R}$.
Then the continuum hypothesis is equivalent to the existence of a map $X\ni x\mapsto\Delta_x\subset X$ such that $\Delta_x$ is countable for every $x$ and for all $x_1,x_2\in X$ 
\begin{equation}\label{cheq}
	x_1\in\Delta_{x_2}~\text{or}~x_2\in\Delta_{x_1}.
\end{equation} 
\end{lemma}
\begin{proof}
The proof takes place in $X^2$.
For any $x\in X$ define sets corresponding to rows and columns in $X^2$:
\begin{equation}
	R_x\ldef\{(x',x)\:|\:x'\in X\},~C_x\ldef\{(x,x')\:|\:x'\in X\}.
\end{equation}
It was shown by \citet{Sierpinski19,Sierpinski34} that the continuum hypothesis is equivalent with the existence of two subsets $R,C\subset X^2$ such that $X^2=R\cup C$ and for every $x\in X$ the sets $R\cap R_x$ and $C\cap C_x$ are countable.\footnote{Validity of the axiom of choice is assumed here.}

Now assume the continuum hypothesis holds and suppose $R,C$ are subsets of $X^2$ as above. 
Because for every $(x_1,x_2)\in X^2$ it holds that it is either an element of $R$ or $C$ it follows that either 
\begin{equation}\label{rowproj}
	x_1\in r_{x_2}\ldef\{x'\in X\:|(x',x_2)\in R_{x_2}\cap R\}
\end{equation}
or
\begin{equation}\label{colproj}
	x_2\in c_{x_1}\ldef\{x'\in X\:|(x_1,x')\in C_{x_1}\cap C\}.
\end{equation}

Now take $\Delta_x\ldef r_x\cup c_x$.
Then $\Delta_x$ is a countable set for every $x$ and from \eqref{rowproj} and \eqref{colproj} it follows that for all $x_1,x_2\in X$ either $x_1\in\Delta_{x_2}$ or $x_2\in\Delta_{x_1}$.

For the converse assume $x\mapsto\Delta_x$ is given.
Define $C\ldef\{(x_1,x_2)\:|\:x_2\in\Delta_{x_1}\}$ and $R\ldef\{(x_1,x_2)\:|\:x_1\in\Delta_{x_2}\}$.
Because $\Delta_x$ is countable, so are $C\cap C_x$ and $R\cap R_x$ and from eq. \eqref{cheq} it follows that $R\cup C=X^2$.
\end{proof}

\begin{proof}[Proof of Theorem \ref{mainthm}]
I start with showing that if the continuum hypothesis holds, then $(\Lambda_{\mathrm{G}},\Sigma_{\mathrm{G}})$ is $\psi$-determinate.
Apply Lemma \ref{CHlemma} to the set $L_1(\h)$ and let $\Delta_{[\psi]}\subset L_1(\h)$ be a countable subset for every $[\psi]$ such that for all $[\psi],[\phi]$ either $[\psi]\in\Delta_{[\phi]}$ or $[\phi]\in\Delta_{[\psi]}$.

Now define canonical $\psi$-ontic subsets $\Lambda^m_{[\psi]}$ by choosing $m$ such that $m([\psi])$ ranges over all elements of $\Delta_{[\psi]}$ but skipping the value $[\psi]$.
From property \eqref{cheq} it follows that for any pair of distinct pure states $[\psi],[\phi]$
\begin{equation}
	\exists n\in\mathbb{N}~\text{s.t.}~
	[\phi]=m^{[\psi]}_n\text{ or }
	[\psi]=m^{[\phi]}_n.
\end{equation}
Thus it follows that
\begin{equation}\label{disj}
	\Lambda^m_{[\psi]}\cap\Lambda^m_{[\phi]}=\varnothing.
\end{equation}

Conversely, suppose there is a map $m$ such that \eqref{disj} holds for all $[\psi],[\phi]$.
This equation holds iff
\begin{equation}
	\Lambda_{[\psi]}\cap\bigcap_{n\in\mathbb{N}}\Lambda^c_{m^{[\phi]}_n}=\varnothing
	\text{ or }
	\Lambda_{[\phi]}\cap\bigcap_{n\in\mathbb{N}}\Lambda^c_{m^{[\psi]}_n}=\varnothing.
\end{equation}
Therefore, the map
\begin{equation}
	[\psi]\mapsto\{[\psi]\}\cup\left\{m^{[\psi]}_n\:\middle|\:n\in\mathbb{N}\right\}\subset L_1(\h)
\end{equation}
satisfies the criteria of the map in Lemma \ref{CHlemma} and the continuum hypothesis holds.
\end{proof}

\subsection{Proof of Theorem \ref{MKCTheorem}}\label{Proof3App}

Let $[\psi]$ and $n$ be given.
Now consider the Euler function $E:[0,1]\to[0,1]$ given by 
\begin{equation}
	E(q)\ldef\prod_{k=1}^\infty(1-q^k).
\end{equation}
For the given $n$ there exists a $q_n>0$ such that $E(q_n)=1-\tfrac{1}{n}$.
Now for every $k\in\mathbb{N}$ choose a context $C_{k,n}$ with a $[\psi_{k,n}]\in C_{k,n}$ such that $\left|\left\langle\psi\middle|\psi_{k,n}\right\rangle\right|^2>1-q_n^k$.
Then define
\begin{equation}
	\Lambda_{[\psi],n}\ldef\bigcap_{k=1}^\infty\Delta_{C_{k,n}}^{\psi_{k,n}},
\end{equation}
so
\begin{equation}
	\mu_{[\psi]}\left(\Lambda_{[\psi],n}\right)
	=\prod_{k=1}^\infty\left|\left\langle\psi\middle|\psi_{k,n}\right\rangle\right|^2
	>\prod_{k=1}^\infty(1-q_n^k)=E(q_n)=1-\frac{1}{n}.
\end{equation}
For every $n$, $\psi_{k,n}$ gets closer and closer to $\psi$ as $k\to\infty$.
So for any $[\phi]\neq[\psi]$ there exists some $\delta>0$ and $K\in\mathbb{N}$ such that $\left|\left\langle\phi\middle|\psi_{k,n}\right\rangle\right|^2<1-\delta$ for all $k>K$. 
Therefore $\mu_{[\phi]}\left(\Lambda_{[\psi],n}\right)=0$. 

\end{appendix}

\end{document}